\UseRawInputEncoding

\documentclass[10pt]{article}

\usepackage{hyperref}
\usepackage{amsthm,amsmath,amssymb}
\usepackage{enumerate}

\usepackage{fullpage}
\newtheorem{theorem}{Theorem}
\newtheorem{lemma}[theorem]{Lemma}
\newtheorem{remark}[theorem]{Remark}

\newtheorem{corollary}[theorem]{Corollary}

\newtheorem{proposition}[theorem]{Proposition}
\newtheorem{example}[theorem]{Example}

\newcommand{\Cov}{\operatorname{Cov}}

\newcommand{\pd }{\mathbf{S}^+}
\newcommand{\lsc }{l.s.c.\ }
\newcommand{\usc }{u.s.c.\ }
\newcommand{\law }{\operatorname{law}}
\newcommand{\psd }{\mathbf{S}_0^+}
\newcommand{\sym }{\mathbf{S}}
\newcommand{\Tr}{\operatorname{Tr}}
\newcommand{\HS}{\operatorname{HS}}
\newcommand{\id}{\operatorname{id}}

\newcommand{\EE}{\mathbb{E}}

\newcommand{\R}{\mathbb{R}}

\title{Entropy Inequalities and Gaussian Comparisons}
\author{Efe Aras and Thomas A.~Courtade\\University of California, Berkeley}
\date{~~}

\begin{document}

\maketitle

\begin{abstract}
We establish a general class of entropy inequalities that take the concise form of Gaussian comparisons.  The main result unifies many classical and recent results, including the Shannon--Stam inequality, the Brunn--Minkowski inequality, the Zamir--Feder inequality, the Brascamp--Lieb and Barthe inequalities, the Anantharam--Jog--Nair inequality, and others. 
\end{abstract}

\section{Introduction}
Entropy inequalities have been a core part of  information theory since its inception; their development  driven largely by the role they serve in impossibility results for coding theorems.  Many basic  inequalities enjoyed by entropy, such as subadditivity, boil down to convexity of the logarithm, and hold in great generality.  Others are decidedly more analytic in nature, and may be regarded as capturing some deeper geometric property of the specific spaces on which they hold.  In the context of Euclidean spaces, a notable example of the latter is the Shannon--Stam entropy power inequality (EPI), stated in Shannon's original 1948 treatise \cite{shannon48} and   later proved by Stam \cite{stam59}.  Another example is the Zamir--Feder inequality \cite{ZamirFeder},  which can be stated as follows: Let $X = (X_1, \dots, X_k)$ be a random vector in $\mathbb{R}^k$  with independent coordinates $(X_i)_{i=1}^k$.   If $Z = (Z_1, \dots, Z_k)$ is a Gaussian vector with independent coordinates $(Z_i)_{i=1}^k$ and entropies satisfying $h(Z_i) = h(X_i)$, $1\leq i \leq k$, then for any linear map $B: \mathbb{R}^k \to \mathbb{R}^n$, we have
\begin{align}
 h( B X)\geq  h( B Z). \label{eq:ZamirFederIneq}
\end{align}

Evidently, \eqref{eq:ZamirFederIneq} takes the form of a Gaussian comparison; so, too, does the EPI.  The goal of this paper is to show that such Gaussian comparisons hold in great generality, thus unifying a large swath of known and new information-theoretic and geometric inequalities.   For example, we'll see that \eqref{eq:ZamirFederIneq} holds   when the $X_i$'s are random vectors of different dimensions, and, in fact, continues to hold even when the independence assumption is suitably relaxed.  As another example, we'll see how the EPI and the Brunn--Minkowski inequality emerge as different endpoints of a suitable Gaussian comparison, thus giving a clear and precise explanation for their formal similarity.

This paper is organized as follows.  Section \ref{sec:MainResult} presents the main result and a few short examples; Section \ref{sec:proofs} is dedicated to the proof.  Sections \ref{sec:multimarginal} and \ref{sec:saddle} give further applications, and Section \ref{sec:closing} delivers closing remarks.

\section{ Main Result} \label{sec:MainResult}
Recall  that a Euclidean space $E$  is a finite-dimensional Hilbert space over the real field, equipped with Lebesgue measure.  For a probability measure $\mu$ on $E$, absolutely continuous with respect to Lebesgue measure, and a random vector $X\sim \mu$, we define the Shannon entropy
$$
h(X) \equiv h(\mu) :=-\int_E \log\left( \frac{d\mu}{dx}\right)d\mu, 
$$
provided the integral exists.  If $\mu$ is not absolutely continuous with respect to Lebesgue measure, we adopt the convention $h(\mu):=-\infty$.  We let $\mathcal{P}(E)$ denote the set of  probability measures on $E$ having finite entropies and second moments.   When there is no cause for ambiguity, we adopt the  hybrid notation where a random vector $X$ and its law $\mu$ are denoted interchangeably.  So, for example, writing $X\in \mathcal{P}(E)$ means that $X$ is a random vector taking values in $E$, having finite entropy and finite second moments.  We let $\mathcal{G}(E)$ denote the subset of $\mathcal{P}(E)$ that consists of Gaussian measures.  

The following notation will be reserved throughout.  We consider a Euclidean space $E_0$ with a fixed orthogonal decomposition $E_0 = \oplus_{i=1}^k E_i$.  There are no constraints on the dimensions of these spaces, other than that they are finite (by definition of Euclidean space), and $\dim(E_0) = \sum_{i=1}^k \dim(E_i)$ (by virtue of the stated decomposition).  We let $\mathbf{d} = (d_j)_{j=1}^m$ be a collection of positive real numbers, and $\mathbf{B}=(B_j)_{j=1}^m$ be a collection of linear maps $B_j : E_0 \to E^j$, with common domain $E_0$ and respective codomains equal to  Euclidean spaces $E^1, \dots, E^m$.  Aside from linearity, no further properties of the maps in $\mathbf{B}$ are assumed. For given random vectors $X_i\in \mathcal{P}(E_i)$, $1\leq i \leq k$, we let $\Pi(X_1, \dots, X_k)$ denote the corresponding set of couplings on $E_0$.   That is, we write $X\in \Pi(X_1, \dots, X_k)$ to indicate that $X$ is a random vector taking values in $E_0$ with 
$$
\pi_{E_i}(X) \overset{law}{=} X_i, ~~1\leq i\leq k,
$$
where $\pi_{E_i} : E_0 \to E_i$ is the canonical projection. 

For $X\in \Pi(X_1,\dots, X_k)$ and $S\subset \{1,\dots,k\}$, we define the {${S}$-correlation}\footnote{The $S$-correlation  $I_S$ seems to have no generally agreed-upon name, and has been called different things in the literature.  Our choice of terminology reflects that of Watanabe \cite{Watanabe}, who used the term {\it total correlation} to describe    $I_S$ when $S=\{1,\dots,k\}$. However, it might also be called $S$-information, to reflect the ``multi-information" terminology preferred by some (see, e.g.,   \cite{CsiszarKorner}).}
$$
I_S(X) := \sum_{i\in S}h(X_i) - h( \pi_{S}(X) ),
$$
where we let $\pi_{S}$ denote the canonical projection of $E_0$ onto $\oplus_{i\in S}E_i$.   To avoid ambiguity, we adopt the convention that $I_{\emptyset}(X) = 0$.  Observe that that $I_S$ is the relative entropy between the law of $\pi_{S}(X)$ and the product of its marginals, so is always nonnegative.  

For a given {constraint function} $\nu : 2^{\{1,\dots, k\}} \to [0,+\infty]$, and  $X_i\in \mathcal{P}(E_i)$, $1\leq i \leq k$, we can now define the  set of 
{\bf correlation-constrained couplings} 
\begin{align*}
&\Pi(X_1, \dots, X_k ; \nu) :=   \big\{ X \in \Pi(X_1, \dots, X_k) :  I_S(X)\leq  \nu(S) \mbox{~for each~} S \in    2^{\{1,\dots, k\} } \big\}.
\end{align*}

With   notation established, our main result is the following. 

\begin{theorem}\label{thm:GaussianComparisonConstrained}  Fix   $(\mathbf{d},\mathbf{B})$ and  $\nu : 2^{\{1,\dots, k\}} \to [0,+\infty]$.  For any  $X_i  \in \mathcal{P}(E_i)$, $1\leq i \leq k$,  there exist  $Z_i \in \mathcal{G}(E_i)$ with $h(Z_i)= h(X_i)$, $1\leq i\leq k$ satisfying
\begin{align}
\max_{X\in \Pi(X_1, \dots, X_k;\nu)}\sum_{j=1}^m d_j h(B_j X) \geq %
\max_{Z\in \Pi(Z_1, \dots, Z_k;\nu)}\sum_{j=1}^m d_j h(B_j Z).
 \label{eq:maxEntComparisonConstrained}
\end{align}
\end{theorem}

\begin{remark}
The special case where $\dim(E_i) = 1$ for all $1\leq i \leq k$ appeared in the preliminary work \cite{ArasCourtadeISIT2021} by the authors.
\end{remark}

Let us give the two brief examples promised in the introduction; further applications are discussed in Sections \ref{sec:multimarginal} and \ref{sec:saddle}.   First, observe that when $m=1$, $\nu\equiv 0$ and $\dim(E_i)=1$ for all $1\leq i \leq k$,  we recover the Zamir--Feder inequality \eqref{eq:ZamirFederIneq}.  Indeed, taking $\nu \equiv 0$ renders the set of couplings equal to the singleton consisting of the independent coupling, and the one-dimensional nature of the $E_i$'s means that the variances of the $Z_i$'s are fully determined by the entropy constraints.  Hence,  it is clear that Theorem \ref{thm:GaussianComparisonConstrained} generalizes the Zamir--Feder inequality \eqref{eq:ZamirFederIneq} in the directions noted in the introduction.  That is, it continues to hold in the case where the $X_i$'s are multidimensional, and when the independence assumption is relaxed in a suitable manner. 

As a second and slightly more substantial example, we explain the connection between the EPI and the Brunn--Minkowski inequality alluded to in the introduction. Denote the {entropy power} of $X\in \mathcal{P}(\mathbb{R}^n)$ by
$$
N(X):=  e^{2 h(X)/n}  .
$$
For a coupling $X=(X_1,X_2)$, note that the  {mutual information} $I(X_1;X_2)$ is equal to  $I_S(X)$ with $S=\{1,2\}$.

\begin{corollary}\label{thm:depEPI}
For any  $X_1,X_2   \in \mathcal{P}(\R^n)$  and   $\zeta \in [0,+\infty]$,  it holds that
\begin{align}
N(X_1) + N(X_2)  + &2 \sqrt{(1 - e^{- 2 \zeta/n }) N(X_1)N(X_2)}   
\leq \!\!\! \max_{ \substack{X_1,X_2 :\\ I(X_1;X_2)\leq \zeta} } \!\!\! N(X_1+X_2)  , \label{eq:depEPI} %
\end{align}
where the maximum is over couplings of $X_1,X_2$ such that  $I(X_1;X_2)\leq \zeta$.  
Equality holds   for  Gaussian $X_1, X_2$ with proportional covariances.
\end{corollary}
\begin{proof}
We apply Theorem \ref{thm:GaussianComparisonConstrained} with  $E_1= E_2=\R^n$ and   $\nu(\{1,2\}) =   \zeta$ to give existence of Gaussian $Z_1,Z_2$ satisfying $N(Z_i)=N(X_i)$ and 
$$
\max_{ \substack{(X_1,X_2) \in \Pi(X_1,X_2) : \\I(X_1;X_2)\leq \zeta} } N(X_1+X_2)
\geq 
 \max_{ \substack{(Z_1,Z_2) \in \Pi(Z_1,Z_2) : \\I(Z_1;Z_2)\leq \zeta} } N(Z_1+Z_2).
$$
Now, suppose  $Z_i\sim N(0,\Sigma_{i})$, $i\in \{1,2\}$  and  
consider the coupling 
$$
Z_1 = \rho \Sigma^{1/2}_{1} \Sigma^{-1/2}_{2} Z_2 + (1-\rho^2)^{1/2} W,
$$
where $W\sim N(0,\Sigma_{1})$ is independent of $Z_2$, and $\rho := (1 - e^{-2 \zeta/n})^{1/2}$.  This    ensures $I(Z_1;Z_2) = \zeta$, and
\begin{align*}
N(Z_1+Z_2) &= (2\pi e) \det( \Sigma_{1}+\Sigma_{2} + \rho \Sigma_{1}^{1/2}\Sigma_{2}^{1/2} +  \rho \Sigma_{2}^{1/2}\Sigma_{1}^{1/2} )^{1/n}\\
&\geq (2\pi e) \left( \det( \Sigma_{1} )^{1/n} + \det( \Sigma_{2} )^{1/n} + 2 \rho  \det( \Sigma^{1/2}_{1} )^{1/n}\det( \Sigma^{1/2}_{2} )^{1/n}\right)\\
&=N(X_1) + N(X_2)  + 2 \sqrt{(1 - e^{-2 \zeta/n}) N(X_1)N(X_2)}, 
\end{align*}
where the inequality follows by Minkowski's determinant inequality.   It is easy to see that we have equality throughout if $X_1, X_2$ are Gaussian with proportional covariances.
\end{proof}

\begin{remark}
Theorem \ref{thm:depEPI} may be considered as an extension of the EPI that holds for  certain dependent random variables; it appeared in the preliminary work \cite{ArasCourtadeISIT2021} by the authors.  We remark that Takano \cite{takano1995inequalities} and Johnson \cite{johnson2004conditional} have established that the EPI holds for dependent random variables which have positively correlated scores.  However, given the different hypotheses,  those results are not directly comparable to Theorem \ref{thm:depEPI}. 
\end{remark}

Now, we observe that the EPI and the Brunn--Minkowski inequality naturally emerge from \eqref{eq:depEPI} by considering the endpoints of independence ($\zeta = 0$) and maximal dependence ($\zeta = +\infty$).  Of course, \eqref{eq:depEPI} also gives a sharp inequality for the whole spectrum of cases in between.

\begin{example}[Shannon--Stam EPI]
Taking $\zeta = 0$ enforces the independent coupling in Theorem \ref{thm:depEPI}, and recovers the  EPI   in its usual form.  For independent $X_1,X_2\in \mathcal{P}(\mathbb{R}^n)$,
\begin{align}
e^{2 h(X_1)/n}+ e^{2 h(X_2)/n}\leq   e^{2 h(X_1+X_2)/n}.\label{eq:EPIstatement}
\end{align}
Hence, Theorem \ref{thm:depEPI} may be regarded as an extension of the  EPI  for certain dependent random variables with a sharp correction term.  
\end{example}

\begin{example}[Brunn--Minkowski inequality]
Taking  $\zeta = +\infty$ in Theorem \ref{thm:depEPI} allows for unconstrained optimization over couplings, giving
$$
e^{h(X_1)/n}+ e^{h(X_2)/n}\leq \sup_{(X_1,X_2) \in \Pi(X_1,X_2)   } e^{h(X_1+X_2)/n},
$$
where we emphasize the change in exponent from $2$ to $1$, relative to \eqref{eq:EPIstatement}.  This may be regarded as an entropic improvement of the Brunn--Minkowski inequality.  Indeed, if  $X_1,X_2$ are uniform on compact subsets $K,L\subset \mathbb{R}^n$, respectively,  we obtain the familiar Brunn--Minkowski inequality
$$
\operatorname{Vol}_n(K)^{1/n} + \operatorname{Vol}_n(L)^{1/n} \leq  \sup_{(X_1,X_2) \in \Pi(X_1,X_2)  } e^{h(X_1+X_2)/n} \leq \operatorname{Vol}_n(K+L)^{1/n},
$$
where $K+L$ denotes the Minkowski sum of $K$ and $L$, and $\operatorname{Vol}_n(\cdot)$ denotes the $n$-dimensional Lebesgue volume.   The last inequality follows since $X_1+X_2$ is supported on the Minkowski sum $K+L$, and hence the entropy is upper bounded by that of the uniform distribution on that set.  
\end{example}

It has long been observed that there is a striking similarity between the Brunn--Minkowski inequality and the EPI (see, e.g., \cite{costa1984similarity} and citing works).  It is well-known that each can be obtained from  convolution inequalities involving R\'{e}nyi entropies (e.g., the sharp Young inequality \cite{ brascamp1976best, lieb1978}, or rearrangement inequalities  \cite{WangMadiman}), when the orders of the involved R\'{e}nyi entropies are taken to the limit $0$ or $1$, respectively.   Quantitatively linking the Brunn--Minkowski and EPI  using only Shannon entropies has  proved elusive, and has  been somewhat of a looming question.  In this sense, Theorem \ref{thm:depEPI} provides an answer. In particular,  the Brunn--Minkowski inequality and EPI are   obtained as logical endpoints of a family of inequalities which involve only Shannon entropies instead of  R\'enyi entropies of varying orders.   In contrast to  derivations involving R\'enyi entropies where summands are always independent (corresponding to the convolution of densities), the  idea here is to allow dependence between the random summands.

We do not tackle the problem of characterizing equality cases in this paper, but we remark that equality is attained in the Brunn--Minkowski inequality when $K,L$ are positive homothetic convex bodies, which highlights that the stated conditions for equality in Theorem \ref{thm:depEPI} are sufficient, but not always necessary.  Indeed, for $X_1,X_2$ equal in distribution,  Cover and Zhang \cite{cover1994maximum} showed  
 $$
h(2 X_1) \leq \max_{ (X_1,X_2) \in \Pi(X_1,X_2)} h(X_1+X_2),
$$
with equality if and only if $X_1$ is log-concave.   We expect that for $\zeta <+\infty$, the only extremizers in Theorem \ref{thm:depEPI} are Gaussian with proportional covariances.  For $\zeta =+\infty$, the resulting entropy inequality is dual to the Pr\'ekopa--Leindler inequality, so   the known equality conditions  \cite{Dubuc} are likely to carry over.  Namely, equality should be attained in this case iff $X_1$ is log-concave and $X_2 = \alpha X_1$ a.s.\ for $\alpha \geq 0$.  

We remark that equality cases for \eqref{eq:maxEntComparisonConstrained} in the special case where $\nu \equiv 0$ follow from the main results in \cite{ArasCourtadeZhang}.

\section{Proof of the Main Result}\label{sec:proofs}
This section is dedicated to the  proof of Theorem \ref{thm:GaussianComparisonConstrained}. There are several preparations to make before starting the proof; this is done in the first subsection.  The second subsection brings everything together to prove an unconstrained version of  Theorem \ref{thm:GaussianComparisonConstrained} where $\nu \equiv +\infty$.  The third and final subsection proves Theorem \ref{thm:GaussianComparisonConstrained} on the basis of its unconstrained variation.

\subsection{Preliminaries}
Here  we quote the preparatory results that we shall need, and the definitions required to state them.  The various results are organized by subsection, and proofs are only given where necessary.

\subsubsection{Some additional notation}
For a Euclidean space $E$, we let $\pd(E)$ denote the set of symmetric positive definite linear operators from $E$ to itself.   That is, $A\in \pd(E)$ means $A = A^T$ and $x^T A x >0 $ for all nonzero $x\in E$.   We let $\psd(E)$ denote the closure of $\pd(E)$, equal to those symmetric matrices which are positive semidefinite.  The set $\sym(E)$ denotes the  matrices which are symmetric.    We let $\langle\cdot,\cdot\rangle_{\HS}$ denote the Hilbert--Schmidt (trace) inner product, and $\|\cdot\|_{\HS}$ denote the induced norm (i.e., the Frobenius norm). 

If $K_i \in \pd(E_i)$, $1\leq i \leq k$, then we let $\Pi(K_1, \dots, K_k)$ denote the subset of $\psd(E_0)$ consisting of those matrices $K$ such that 
$$
\pi_{E_i} K \pi_{E_i}^T = K_i, ~~~1\leq i \leq k.
$$
Note that this overloaded notation is consistent with our notation for couplings.  Indeed, if $X_i \sim N(0,K_i)$, $1\leq i \leq k$, then  $X \sim N(0,K)$ is a coupling in $\Pi(X_1, \dots, X_k)$ if and only if $K \in \Pi(K_1, \dots, K_k)$.

 If $A_i : E_i \to E_i$, $1\leq i \leq k$, are linear maps, then we write the block-diagonal matrix
$$
A = \operatorname{diag}(A_1, \dots, A_k)
$$
to denote the operator direct sum $A =    \oplus_{i=1}^k A_i  : E_0 \to E_0$.  For a set $V$, we  let $\id_{V}: V\to V$ denote the identity map from $V$ to itself.  So, for instance, we have  $\id_{E_0} = \oplus_{i=1}^k \id_{E_i} \equiv  \operatorname{diag}(\id_{E_1}, \dots, \id_{E_k})$.  

\subsubsection{The entropic forward-reverse Brascamp--Lieb inequalities}
Define
$$
D_g(\mathbf{c},\mathbf{d},\mathbf{B}) := \sup_{Z_i \in \mathcal{G}(E_i),1\leq i \leq k }\left( \sum_{i=1}^k c_i h(Z_i) - \max_{Z\in \Pi(Z_1, \dots, Z_k)}\sum_{j=1}^m d_j h(B_j Z) \right),
$$
The following is a main result of \cite{CourtadeLiu21}, when stated in terms of entropies.
\begin{theorem}\label{thm:FRBLentropy}  Fix a datum $(\mathbf{c},\mathbf{d},\mathbf{B})$.  For random vectors $X_i   \in \mathcal{P}(E_i)$, $1\leq i \leq k$, we have
\begin{align}
\sum_{i=1}^k c_i h(X_i) \leq \max_{X\in \Pi(X_1, \dots, X_k)}\sum_{j=1}^m d_j h(B_j X) + D_g(\mathbf{c},\mathbf{d},\mathbf{B}). \label{eq:MainEntropyCouplingInequality}
\end{align}
Moreover, the constant $D_g(\mathbf{c},\mathbf{d},\mathbf{B})$ is finite if and only if the following two conditions hold.
\begin{enumerate}[(i)]
\item {\bf Scaling condition:} It holds that
\begin{align}
\sum_{i=1}^k c_i \dim(E_i) = \sum_{j=1}^m d_j \dim(E^j). \label{eq:ScalingCond}
\end{align}
\item{\bf Dimension condition:}    For all subspaces $T_i \subset E_i$, $1\leq i \leq k$, 
 \begin{align}
\sum_{i=1}^k c_i \dim(T_i ) \leq \sum_{j=1}^m d_j \dim(B_j T),\hspace{5mm}\mbox{where $T = \oplus_{i=1}^k T_i$.} \label{eq:DimCond}
\end{align}
\end{enumerate}
\end{theorem}

A datum $(\mathbf{c},\mathbf{d},\mathbf{B})$ is said to be {\bf extremizable} if $D(\mathbf{c},\mathbf{d},\mathbf{B})<\infty$ and there exist  $X_i \in \mathcal{P}(E_i)$, $1\leq i \leq k$ which attain equality in \eqref{eq:MainEntropyCouplingInequality}. Likewise, a  datum $(\mathbf{c},\mathbf{d},\mathbf{B})$  is said to be 
 {\bf Gaussian-extremizable} if  there exist Gaussian $X_i \in \mathcal{G}(E_i)$, $1\leq i \leq k$ which attain equality in \eqref{eq:MainEntropyCouplingInequality}.  Necessary and sufficient conditions for Gaussian-extremizability of a datum $(\mathbf{c},\mathbf{d},\mathbf{B})$ can be found in \cite{CourtadeLiu21}.  Clearly Gaussian-extremizability implies extremizability on account of Theorem \ref{thm:FRBLentropy}.  We shall need the converse, which was not proved in \cite{CourtadeLiu21}.
\begin{theorem}\label{thm:extImpliesGext}
If a datum $(\mathbf{c},\mathbf{d},\mathbf{B})$ is extremizable, then it is Gaussian-extremizable. 
\end{theorem}
The proof follows a doubling argument similar to what appears \cite[Proof of Theorem 8]{liu2018forward}.  We will need the following Lemma.
\begin{lemma}\label{lem:W2ConvergenceCovariance}
For each $1\leq i \leq k$, let  $Z_i\sim N(0,K_i)$ and let $(X_{n,i})_{n\geq 1}$ be a sequence of zero-mean random vectors satisfying 
$$\lim_{n\to\infty} W_2(X_{n,i}, Z_i)= 0,$$
where $W_2: \mathcal{P}(E_i)\times  \mathcal{P}(E_i)\to \mathbb{R}$ is the Wasserstein distance of order 2.  
  For any $K\in \Pi(K_1, \dots, K_k)$, there exists a sequence of couplings $X_n \in \Pi(X_{n,1},\dots, X_{n,k})$, $n\geq 1$ such that $\|\Cov(X_n) - K\|_{\HS}\to 0$.
\end{lemma}
\begin{proof}
Let $Z\sim N(0,K)$, and observe that $Z \in \Pi(Z_1, \dots, Z_k)$.  Let $T_{n,i}$ be the optimal transport  map sending $N(0,K_i)$ to $\law(X_{n,i})$ (see, e.g., \cite{villani2003topics}).  Then   $X_n = (T_{n,1}(Z_1), \dots,  T_{n,k}(Z_k))  \in \Pi(X_{n,1},\dots, X_{n,k})$ satisfies 
\begin{align*}
T_{n,i}(Z_{i})T_{n,i'}(Z_{i'})^T - Z_i Z_{i'}^T &=    Z_i (T_{n,i'}(Z_{i'}) - Z_{i'}   )^T +  (T_{n,i}(Z_{i}) - Z_{i}) Z_{i'}^T \\
  &\phantom{=}+    (T_{n,i}(Z_{i}) - Z_{i})  (T_{n,i'}(Z_{i'}) - Z_{i'}   )^T .
\end{align*}
Taking expectations of both sides and applying Cauchy--Schwarz, we conclude
$$
\|\Cov(X_n) - K\|_{\HS} \to 0
$$
since $\EE|T_{n,i}(Z_{i}) - Z_{i}|^2 = W_2(X_{n,i}, Z_i)^2 \to 0$ for each $1\leq i\leq k$. 
\end{proof}

\begin{proof}[Proof of Theorem \ref{thm:extImpliesGext}] The approach will be to show that extremizers are closed under convolutions, and apply the entropic central limit theorem.  Toward this end, let $X_i \sim \mu_i \in \mathcal{P}(E_i)$ be independent of $Y_i \sim \nu_i \in \mathcal{P}(E_i)$, $1\leq i \leq k$, both assumed to be extremal in \eqref{eq:MainEntropyCouplingInequality}.  Define
$$
Z_i^+ := X_i + Y_i, \hspace{5mm} Z_i^- := X_i - Y_i, \hspace{5mm}1\leq i \leq k,
$$
and let
  $$
  Z^+ \in \arg\max_{Z\in \Pi(Z_1^+, \dots, Z_k^+)} \sum_{j=1}^m d_j h(B_j Z).
 $$
 Let $Z_i^-|z_i^+$ denote the random variable $Z_i^-$ conditioned on $\{Z_i^+= z_i^+\}$, which has law in $\mathcal{P}(E_i)$ for $\law(Z_i^+)$-a.e.~$z_i^+\in E_i$ by disintegration.  Next, for $z^+ = (z_1^+, \dots,  z_k^+)\in E_0$, let
 $$
 Z^-|z^+ \in  \arg\max_{Z\in \Pi(Z_1^-|z_1^+, \dots, Z_k^-|z_k^+)} \sum_{j=1}^m d_j h(B_j Z).
 $$
We can assume these couplings are  such  that $z^+\mapsto \law( Z^-|z^+)$ is Borel measurable (i.e., $\law( Z^-|z^+)$ is a regular conditional probability). This can be justified by  measurable selection theorems, as in \cite[Cor. 5.22]{villani2008} and \cite[p. 42]{liu2017ITperspectiveBL}.  
  With this assumption, definitions imply
 \begin{align*}
  \sum_{i=1}^k c_i h(Z^+_i) &\leq \sum_{j=1}^m d_j h(B_j Z^+)  + D(\mathbf{c},\mathbf{d},\mathbf{B})\\
 \sum_{i=1}^k c_i h(Z^-_i |  z_i^+ ) &\leq \sum_{j=1}^m d_j h(B_j Z^-|  z^+)  + D(\mathbf{c},\mathbf{d},\mathbf{B}),
 \end{align*}
 where the latter holds for $\law(Z^+)$-a.e.~$z^+$.  Integrating the second inequality against the distribution of $Z^+$ gives the inequality for conditional entropies:
  \begin{align*}
 \sum_{i=1}^k c_i h(Z^-_i |  Z_i^+ ) &\leq \sum_{j=1}^m d_j h(B_j Z^-|  Z^+)  + D(\mathbf{c},\mathbf{d},\mathbf{B})\\
 &\leq \sum_{j=1}^m d_j h(B_j Z^-|  B_j Z^+)  + D(\mathbf{c},\mathbf{d},\mathbf{B}),
 \end{align*}
 where the second inequality follows since conditioning reduces entropy.  Now, define
 $$
 X = \frac{1}{2}\left( Z^+ + (Z^-|Z^+)   \right) , \hspace{5mm} Y = \frac{1}{2}\left( Z^+ - (Z^-|Z^+)   \right). 
 $$
 Observe that $X\in \Pi(X_1, \dots, X_k)$ and $Y\in \Pi(Y_1, \dots, Y_k)$.  So, using the above inequalities and definitions, we have
\begin{align*}
2 D(\mathbf{c},\mathbf{d},\mathbf{B}) &\leq   \sum_{i=1}^k c_i h(X_i,Y_i)    - \sum_{j=1}^m d_j h(B_j X )  - \sum_{j=1}^m d_j h(B_j Y )     \\
&\leq   \sum_{i=1}^k c_i h(X_i,Y_i)    - \sum_{j=1}^m d_j h(B_j X,B_j Y)   \\
&= \sum_{i=1}^k c_i h(Z^+_i) + \sum_{i=1}^k c_i h(Z^-_i | Z_i^+)  \\
&\phantom{=}- \sum_{j=1}^m d_j h(B_j Z^+)  - \sum_{j=1}^m d_j h(B_j Z^-|B_j Z^+) \\
&\leq 2 D(\mathbf{c},\mathbf{d},\mathbf{B}) 
\end{align*}
Thus, we conclude 
$$
\sum_{i=1}^k c_i h(Z^+_i) = \sum_{j=1}^m d_j h(B_j Z^+) +  D(\mathbf{c},\mathbf{d},\mathbf{B}), 
$$
showing that $Z_i^+ \sim \mu_i*\nu_i \in \mathcal{P}(E_i)$, $1\leq i\leq k$ are extremal in \eqref{eq:MainEntropyCouplingInequality} as desired.  The scaling condition \eqref{eq:ScalingCond} is necessary for $D(\mathbf{c},\mathbf{d},\mathbf{B})<\infty$, so it follows by induction and scale invariance that, for every $n\geq 1$,   marginally specified $(Z_{n,i})_{i=1}^k$ are extremal in \eqref{eq:MainEntropyCouplingInequality}, where
$$
Z_{n,i}:=\frac{1}{\sqrt{n}}\sum_{\ell=1}^n (X_{\ell,i}-\EE[X_i]),  
$$
and $(X_{\ell,i})_{\ell\geq 1}$ are i.i.d.\ copies of $X_i$.

 Define $K_i = \Cov(X_i)$ (which is positive definite since $h(X_i)$ is finite), and fix any $K \in \Pi(K_1, \dots, K_k)$.  For any $\epsilon>0$, Lemma \ref{lem:W2ConvergenceCovariance} together with the central limit theorem  for $W_2$ implies there exists $N \geq 1$ and a coupling $Z_N \in \Pi(Z_{N,1},\dots, Z_{N,k})$ such that $\|\Cov(Z_N)-K\|_{\HS}<\epsilon$.  Letting $Z_N^{(n)}$ denote the standardized sum of $n$ i.i.d.\ copies of $Z_N$, we have $Z^{(n)}_N \in \Pi(Z_{nN,1},\dots, Z_{nN,k})$ for each $n\geq 1$.  Thus, by the entropic central limit theorem \cite{barronCLT, CarlenSoffer}, we have
\begin{align*}
\limsup_{n\to \infty} \max_{Z_n \in \Pi(Z_{n,1},\dots, Z_{n,k})} \sum_{j=1}^m d_j h(B_j Z_n) &\geq \lim_{n\to\infty}  \sum_{j=1}^m d_j h(B_j Z^{(n)}_N  )=\sum_{j=1}^m d_j h(B_j Z^{*}_N  )
\end{align*}
where $Z^{*}_N\sim N(0,\Cov(Z_N))$.  Our arbitrary choice of $K$ and $\epsilon$ together with continuity  of determinants implies 
\begin{align*}
&\limsup_{n\to\infty} \max_{Z_n \in \Pi(Z_{n,1},\dots, Z_{n,k})} \sum_{j=1}^m d_j h(B_j Z_n) \geq\max_{K \in \Pi(K_1, \dots, K_k)  }\sum_{j=1}^m  \frac{d_j}{2}\log \left(  (2\pi e)^{\dim(E^j)} \det( B_j K B_j^T )\right).
\end{align*}
Invoking the entropic central limit theorem, and using the fact that   $(Z_{n,i})_{i=1}^k$ are extremal in \eqref{eq:MainEntropyCouplingInequality} for each $n\geq 1$, we   conclude
\begin{align*}
\sum_{i=1}^k  \frac{c_i}{2}\log \left(  (2\pi e)^{\dim(E_i)} \det( K_i )\right) &= \lim_{n\to\infty}\sum_{i=1}^k c_i h(Z_{n,i})\\
&=\lim_{n\to\infty}  \max_{Z_n \in \Pi(Z_{n,1},\dots, Z_{n,k})} \sum_{j=1}^m d_j h(B_j Z_n)  + D(\mathbf{c},\mathbf{d},\mathbf{B})\\
&\geq\max_{K \in \Pi(K_1, \dots, K_k)  }\sum_{j=1}^m  \frac{d_j}{2}\log \left(  (2\pi e)^{\dim(E^j)} \det( B_j K B_j^T )\right)+ D(\mathbf{c},\mathbf{d},\mathbf{B}).
\end{align*} 
Thus, by definitions, we have equality throughout, and  $(\mathbf{c},\mathbf{d},\mathbf{B})$ is Gaussian-extremizable.
\end{proof}

\subsubsection{Properties of the max-entropy term}
Let us  briefly   make a few technical observations related to the max-entropy quantity that appears in \eqref{eq:MainEntropyCouplingInequality}.  First, we quote a technical lemma that will be needed several times.  A proof can be found in \cite[Lemma A2]{liu2018forward}.
\begin{lemma} \label{lem:WeakSemicontH} Let $(\mu_n)_{n\geq 1} \subset\mathcal{P}(E)$ converge in distribution to $\mu$.  If $\sup_{n\geq 1}\int_E |x|^2 d\mu_n  < \infty$, then 
$$
\limsup_{n\to\infty}h(\mu_n) \leq h(\mu).  
$$
\end{lemma}

Now, we point out that the max-entropy term is well-defined as a maximum. 
\begin{proposition}\label{prop:MaxEntropyCouplingExists} Fix $(\mathbf{d},\mathbf{B})$ and $X_i\in \mathcal{P}(E_i)$, $1\leq i \leq k$.  The function
$$
X \in \Pi(X_1,\dots, X_k)\longmapsto \sum_{j=1}^m d_j h(B_j X)
$$
achieves a maximum at some $X^* \in \Pi(X_1,\dots, X_k)$.  Moreover, if each $X_i$ is Gaussian, then $X^*$ is  Gaussian.%
\end{proposition}
\begin{proof}
We have  $\sup_{X \in \Pi(X_1,\dots, X_k)}\EE|B_j X|^2 <  \infty$ for each $1\leq j \leq m$ since each $X_i$ has bounded second moments.  The second moment constraint also implies $\Pi(X_1,\dots, X_k)$ is tight, and it is easily checked to be closed in the weak topology.  Thus, Prokhorov's theorem ensures $\Pi(X_1,\dots, X_k)$ is sequentially compact.  So,  if $(X^{(n)})_{n\geq 1}\subset \Pi(X_1,\dots, X_k)$ is such that 
$$
\lim_{n\to\infty}\sum_{j=1}^m d_j h(B_j X^{(n)}) = \sup_{X \in \Pi(X_1,\dots, X_k)} \sum_{j=1}^m d_j h(B_j X),
$$
we can assume $X^{(n)}\to X^* \in \Pi(X_1,\dots, X_k)$ weakly, by  passing to a subsequence if necessary.  This implies $B_j X^{(n)}\to B_jX^*$ weakly for each $1\leq j\leq m$.  The first claim follows by an application of Lemma \ref{lem:WeakSemicontH}.

The second claim now follows from the first, together with the fact that Gaussians maximize entropy under a covariance constraint.
\end{proof}

 Next, if $X_i \sim N(0,K_i)$ for $K_i \in \pd(E_i)$, $1\leq i \leq k$, then the entropy maximization in \eqref{eq:MainEntropyCouplingInequality} is equivalent to the following optimization problem
\begin{align}
   (K_i)_{i=1}^k \mapsto  \max_{K \in \Pi(K_1, \dots, K_k) }   \sum_{j=1}^md_j \log \det(B_j K B_j^T).  \label{eq:maxCouplingsContPro}
\end{align}
 This maximization enjoys a certain strong duality property, which is a consequence of the Fenchel--Rockafellar theorem.  The following can be found in \cite[Theorem 2.8]{CourtadeLiu21}. 
 \begin{theorem}\label{thm:FRdualQuadraticForms} Fix $(\mathbf{d},\mathbf{B})$.  For any $K_i \in \pd(E_i)$, $1\leq i\leq k$, it holds that
\begin{align}
&\max_{K \in \Pi(K_1, \dots, K_k) }\sum_{j=1}^m d_j     \log \det \left( B_j   K   B_j^T \right) + \sum_{j=1}^m d_j  \dim(E^j) \notag \\
&=\inf_{(U_i,V_j)_{1\leq i\leq k, 1\leq j \leq m}}  \left( 
 \sum_{i=1}^k    \langle U_i, K_i\rangle_{\HS} - \sum_{j=1}^m d_j \log \det V_j\right) , \label{FenchelMaxCouplingIntro}
\end{align}
where the infimum is over $U_i\in \pd(E_i),1\leq i\leq k$ and $V_j\in \pd(E^j), 1\leq j\leq m$ satisfying  
\begin{align}
\sum_{j=1}^m d_j B_j^T V_j B_j \leq \operatorname{diag}(  U_1, \dots,   U_k). \label{eq:MinMaxOperatorHypothesisIntro}
\end{align}
\end{theorem}

\begin{corollary}\label{cor:ContinuityOfMaxDet} The function in \eqref{eq:maxCouplingsContPro} is continuous on $\prod_{i=1}^k \pd(E_i)$. 
\end{corollary}
\begin{proof}
By \eqref{FenchelMaxCouplingIntro}, we see that the mapping in \eqref{eq:maxCouplingsContPro} is a pointwise infimum of functions that are affine in $(K_i)_{i=1}^k$, so it follows that it is upper semi-continuous on $\prod_{i=1}^k \pd(E_i)$.  On the other hand,  each $K\in \Pi(K_1, \dots, K_k)$ can be factored as  
$K= K^{1/2}_d \Sigma  K^{1/2}_d$, for $K^{1/2}_d := \operatorname{diag}(K^{1/2}_1, \dots, K^{1/2}_k)$ and  $\Sigma\in \Pi(\id_{E_1}, \dots, \id_{E_k})$.  Since the map $K_i \mapsto K_i^{1/2}$ is continuous on $ \pd(E_i)$, and determinants are also continuous,  it follows that \eqref{eq:maxCouplingsContPro} is a pointwise supremum of continuous functions.  As such, it is  lower semi-continuous, completing the proof.  
\end{proof}

\subsubsection{Convexity properties of $D_g(\mathbf{c},\mathbf{d},\mathbf{B})$} \label{sec:GeoConvex}

For  $(\mathbf{d},\mathbf{B})$ fixed,  define the   function $F: \mathbb{R}^k \times\prod_{i=1}^k\pd (E_i) \to \mathbb{R}\cup\{-\infty\}$ via
\begin{align*}
F\left(\mathbf{c},  (K_i)_{i=1}^k\right) &:= \max_{K \in \Pi(K_1, \dots, K_k)}   \sum_{j=1}^md_j \log \det(B_j K B_j^T)-\sum_{i=1}^k c_i \log \det(K_i)  .
\end{align*}
The motivation for the above definition is that we have
\begin{align}
-2 D_g(\mathbf{c},\mathbf{d},\mathbf{B}) =  \inf_{ (K_i)_{i=1}^k \in \prod_{i=1}^k\pd (E_i)}   F\left(\mathbf{c}, (K_i)_{i=1}^k\right)\label{eq:DgFromF}
\end{align}
by definition of $D_g(\mathbf{c},\mathbf{d},\mathbf{B})$ and the fact that the scaling condition \eqref{eq:ScalingCond} is a necessary condition for finiteness of $D_g(\mathbf{c},\mathbf{d},\mathbf{B})$.  The optimization problem above is not convex in the $K_i$'s, however it is \emph{geodesically-convex}.  This property was mentioned to the second named author by Jingbo~Liu   in a discussion of the geodesically convex  formulation of the Brascamp--Lieb constant \cite{liu2019private,Sra2018}.  We assume the following argument, which extends that for the Brascamp--Lieb constant, was what he had in mind, so we credit the observation to him.

Let us first explain what is meant by geodesic convexity.  Given a metric space $(M,\rho)$ and points $x,y\in M$, a geodesic is a path $\gamma : [0,1] \to M$ with $\gamma(0)=x$, $\gamma(1)=y$ and
$$
d_M\left( \gamma(t_1),\gamma(t_2) \right) = |t_1-t_2| \rho(x,y), \hspace{5mm}\forall t_1,t_2\in [0,1]. 
$$
A function $f:M\to \mathbb{R}$ is   {geodesically-convex} if, for any geodesic $\gamma$, 
$$
f(\gamma(t)) \leq t f(\gamma(0)) + (1-t) f(\gamma(1)), \hspace{5mm}\forall t\in [0,1]. 
$$
The space $(M,\rho)$ is a unique geodesic metric space if every two points $x,y\in M$ are joined by a unique geodesic.

This is relevant to us as follows. For a Euclidean space $E$, the space $(\pd (E),\delta_2)$ is a unique geodesic metric space, where for $A,B\in \pd (E)$, 
$$
t\in [0,1] \mapsto  A\#_tB: = A^{1/2}(A^{-1/2}B A^{-1/2})^t A^{1/2} 
$$
is the unique geodesic joining $A$ and $B$ with respect to the metric
$$
\delta_2(A,B):= \left( \sum_{i=1}^{\dim(E)} \log(\lambda_i(A^{-1}B))^2 \right)^{1/2}  .
$$
The matrix $A\#B := A\#_{1/2} B$ is  referred to as  the {geometric mean} of $A,B\in \pd (E)$.  

The topology on $\pd (E)$ generated by $\delta_2$ is the usual one, in the sense that $\delta_2(A_n,A)\to 0$ if and only if $\|A_n - A\|_{\HS}\to 0$.  Hence, there are no subtleties with regards to the notions of continuity,  etc. In particular, if $f:\pd (E)\to \mathbb{R}$ is continuous and {geodesically midpoint-convex}, i.e.,  
$$
f(A\#B) \leq \frac{1}{2} f(A) + \frac{1}{2} f(B), \hspace{5mm}A,B\in \pd (E), 
$$
then it is geodesically convex.

\begin{theorem} \label{thm:FunctionalPropertiesDg}Fix $(\mathbf{d},\mathbf{B})$.  
\begin{enumerate}[(i)]
\item The function $\mathbf{c} \mapsto D_g(\mathbf{c},\mathbf{d},\mathbf{B})$ is   convex and lower semi-continuous. 
\item For fixed $\mathbf{c}$, the function $(K_i)_{i=1}^k  \mapsto   F\left(\mathbf{c},   (K_i)_{i=1}^k\right)$
is geodesically-convex and continuous on $\prod_{i=1}^k\pd (E_i)$.
\end{enumerate}
\end{theorem}
\begin{remark}
It may be the case that $D_g(\mathbf{c},\mathbf{d},\mathbf{B})=+\infty$ for each $\mathbf{c}$, e.g., if some $B_j$ fails to be surjective.
\end{remark}

Before the proof, we recall a few basic facts about the geometric mean $A\#B$.  A linear transformation $\Phi : \sym(E)\to \sym(E')$  is said to be \emph{positive} if it sends $\pd(E)$ into $\pd(E')$.
\begin{proposition}\label{prop:GeoMeanProperties}
Let $E,E'$ be Euclidean spaces.  For $A_1,A_2,B_1,B_2 \in \pd (E)$, the following hold.
\begin{enumerate}[(i)] 
\item (Monotone Property) If $A_1\geq B_1$ and $A_2\geq B_2$, then   $(A_1\#A_2)\geq (B_2\#B_2)$.
\item (Cauchy--Schwarz) We have
$$
\langle A_1,B_1  \rangle_{\HS}+ \langle A_2,B_2  \rangle_{\HS} \geq 2 \langle (A_1\#A_2), (B_1\#B_2)\rangle_{\HS}.
$$
\item (Ando's inequality) If $\Phi : \sym (E)\to \sym (E')$ is a positive linear map, then
$$
\Phi(A_1\#A_2) \leq \Phi(A_1)\#\Phi(A_2).%
$$
\item (Geodesic linearity of $\log\det$) It holds that
$$
\log\det(A_1 \# A_2) = \frac{1}{2}\log\det(A_1) +  \frac{1}{2}\log\det(A_2). 
$$
\end{enumerate}
\end{proposition}
\begin{proof}
The monotonicity property can be found, e.g., in \cite[p.~802]{Lawson2001}.  By a change of variables using \cite[Lem.~3.1]{Lawson2001}  and \cite[Cor.~2.1(ii)]{ando79}, it suffices to prove (ii) under the assumption that $B_1 = \id_E$.  In particular, Cauchy--Schwarz gives
\begin{align*}
 |\langle (A_1\#A_2), (\id_E\#B_2)\rangle_{\HS} |^2&= 
  |\langle (A_2^{-1/2} A_1 A_2^{-1/2})^{1/2}A_2^{1/2} , A_2^{1/2}B_2^{1/2}\rangle_{\HS} |^2\\
&\leq   \| (A_2^{-1/2} A_1 A_2^{-1/2})^{1/2}A_2^{1/2} \|_{\HS}  \| A_2^{1/2}B_2^{1/2}\|_{\HS} \\
&=\langle A_1, \id_E \rangle_{\HS}  \langle A_2, B_2 \rangle_{\HS} .
\end{align*}
Thus, the claim follows by taking square roots of both sides and invoking the AM-GM inequality $\sqrt{ab}\leq (a+b)/2$ for $a,b\geq 0$.  Ando's inequality can be found in \cite[Thm.~3(i)]{ando79}.   Claim (iv) is trivial. 
\end{proof}

Theorem \ref{thm:FunctionalPropertiesDg} now follows as an easy consequence of the above properties and Theorem \ref{thm:FRdualQuadraticForms}. 
\begin{proof}[Proof of Theorem \ref{thm:FunctionalPropertiesDg}]  Claim (i) follows immediately from \eqref{eq:DgFromF}, since $-D_g(\mathbf{c},\mathbf{d},\mathbf{B})$ is a pointwise infimum of functions that are affine  in $\mathbf{c}$.

To prove (ii), we note that geodesic-linearity of $\log\det$ implies it suffices to show geodesic midpoint-convexity of the continuous (by Corollary  \ref{cor:ContinuityOfMaxDet}) function
\begin{align}
  (K_i)_{i=1}^k \mapsto  \max_{K \in \Pi(K_1, \dots, K_k) }   \sum_{j=1}^md_j \log \det(B_j K B_j^T). \label{eq:maxCouplingsCont}
\end{align}
Invoking Theorem \ref{thm:FRdualQuadraticForms}, this is the same as establishing geodesic-convexity of  
\begin{align}
 (K_i)_{i=1}^k \mapsto \inf_{(U_i,V_j)_{1\leq i\leq k, 1\leq j \leq m}}  \left( 
 \sum_{i=1}^k    \langle U_i, K_i\rangle_{\HS} - \sum_{j=1}^m d_j \log \det V_j\right) , \label{FenchelMaxCouplingGC}
\end{align}
where the infimum is over $U_i\in \pd(E_i),1\leq i\leq k$ and $V_j\in \pd(E^j), 1\leq j\leq m$ satisfying  
\begin{align}
 \operatorname{diag}(  U_1, \dots,   U_k) \geq \sum_{j=1}^m d_j B_j^T V_j B_j . \label{eq:MinMaxOperatorHypothesisGC}
\end{align}
For $\ell\in \{1,2\}$, let $U^{(\ell)}_i\in \pd(E_i),1\leq i\leq k$ and $V^{(\ell)}_j\in \pd(E^j), 1\leq j\leq m$ satisfy \eqref{eq:MinMaxOperatorHypothesisGC}   {with strict inequality}.  As such, there exists $\epsilon>0$ sufficiently small such that 
\begin{align*}
 \operatorname{diag}(  U^{(\ell)}_1, \dots,   U^{(\ell)}_k) \geq &\sum_{j=1}^m d_j B_j^T V^{(\ell)}_j B_j +\epsilon  \sum_{j=1}^m \Tr(V^{(\ell)}_j) \id_{E_0}, \hspace{5mm}\ell\in \{1,2\}.%
\end{align*}
Define the positive  linear map $\Phi :\pd(E^0) \to \pd(E_0)$ via 
$$
\Phi(V) := \sum_{j=1}^m d_j B_j^T \pi_{E^j}V\pi_{E_j}^T B_j  + \epsilon \Tr(V) \id_{E_0},\hspace{5mm}V\in \pd(E^0).
$$
By the monotone property and Ando's inequality in Proposition \ref{prop:GeoMeanProperties},  
\begin{align*}
\operatorname{diag}(  U^{(1)}_1\#U^{(2)}_1, \dots,   U^{(1)}_k\#U^{(2)}_k) &\geq \Phi\left( \operatorname{diag}(  V^{(1)}_1, \dots,   V^{(1)}_m)   \right) \#\Phi\left( \operatorname{diag}(  V^{(2)}_1, \dots,   V^{(2)}_m)   \right) \\
&\geq \Phi\left( \operatorname{diag}(  V^{(1)}_1\#V^{(2)}_1 , \dots,   V^{(1)}_m\#V^{(2)}_m)   \right) \geq \sum_{j=1}^m d_j B_j^T (V^{(1)}_j\#V^{(2)}_j) B_j .
\end{align*}
In particular, $(U^{(1)}_i\# U^{(2)}_i)\in \pd(E_i),1\leq i\leq k$ and $(V^{(1)}_j\# V^{(2)}_j)\in \pd(E^j)$, $1\leq j\leq m$ satisfy \eqref{eq:MinMaxOperatorHypothesisGC}.  Therefore, let $ (K^{(\ell)}_i)_{i=1}^k\in \prod_{i=1}^k\pd (E_i)$ and use Proposition \ref{prop:GeoMeanProperties} to write 
\begin{align*}
&\frac{1}{2}\sum_{\ell\in \{1,2\}} \left( 
 \sum_{i=1}^k    \langle U^{(\ell)}_i, K^{(\ell)}_i\rangle_{\HS} - \sum_{j=1}^m d_j \log \det V^{(\ell)}_j\right)\\
 &\geq  
 \sum_{i=1}^k    \langle ( U^{(1)}_i\#U^{(2) }_i ) , ( K^{(1)}_i\#K^{(2) }_i )\rangle_{\HS} - \sum_{j=1}^m d_j \log \det (V^{(1)}_j \# V^{(2)}_j) \\
 &\geq  
 \inf_{(U_i,V_j)_{1\leq i\leq k, 1\leq j \leq m}}  \left( 
 \sum_{i=1}^k    \langle U_i, ( K^{(1)}_i\#K^{(2) }_i ) \rangle_{\HS} - \sum_{j=1}^m d_j \log \det V_j\right) .
\end{align*}
By continuity of the objective in \eqref{FenchelMaxCouplingGC} with respect to the $U_i$'s, the value of the infimum in  \eqref{FenchelMaxCouplingGC} remains unchanged if we take infimum over $U_i$'s and $V_j$'s satisfying   \eqref{eq:MinMaxOperatorHypothesisGC} with strict inequality.  Hence, by the arbitrary choice of $U^{(\ell)}_i\in \pd(E_i),1\leq i\leq k$ and $V^{(\ell)}_j\in \pd(E^j), 1\leq j\leq m$ subject to \eqref{eq:MinMaxOperatorHypothesisGC} with strict inequality, geodesic midpoint-convexity of \eqref{FenchelMaxCouplingGC} is proved.  
\end{proof}

\subsubsection{Sion's theorem for geodesic metric spaces}

We will need the following version of Sion's minimax theorem, found in  \cite{Zhang2022}.  
 \begin{theorem}[Sion's theorem in geodesic metric spaces]\label{thm:SionGeodesic}
Let $(M,d_M)$ and $(N,d_N)$ be finite-dimensional unique geodesic metric spaces. Suppose $\mathcal{X}\subset M$ is a compact and geodesically convex set, $\mathcal{Y} \subset N$ is a geodesically convex set. If following conditions hold for $f : \mathcal{X} \times \mathcal{Y} \to \mathbb{R}$:
\begin{enumerate}[1.]
\item  $f (\cdot, y)$ is geodesically-convex and \lsc for each $y\in \mathcal{Y}$; 
\item $f (x, \cdot)$ is geodesically-concave and \usc for each $x\in \mathcal{X}$, 
\end{enumerate}
then
$$
\min_{x\in \mathcal{X}} \sup_{y\in \mathcal{Y}}  f(x,y) =  \sup_{y\in \mathcal{Y}} \min_{x\in \mathcal{X}} f(x,y).
$$
 \end{theorem}

\subsection{Unconstrained comparisons}

With all the pieces in place, we can take a big step toward proving Theorem \ref{thm:GaussianComparisonConstrained} by first establishing the result in the unconstrained case.  Namely, the goal of this section is to prove the following.

\begin{theorem}\label{thm:GaussianComparisons}  Fix   $(\mathbf{d},\mathbf{B})$.  For any  $X_i   \in \mathcal{P}(E_i)$, $1\leq i \leq k$,    there exist   $Z_i \in \mathcal{G}(E_i)$ with $h(Z_i)= h(X_i)$ for $1\leq i\leq k$ such that 
\begin{align}
\max_{X\in \Pi(X_1, \dots, X_k)}\sum_{j=1}^m d_j h(B_j X) \geq  \max_{Z\in \Pi(Z_1, \dots, Z_k)}\sum_{j=1}^m d_j h(B_j Z). 
 \label{eq:maxEntComparison}
\end{align}
\end{theorem}
\begin{remark}
It is a part of the theorem that each maximum is attained.
\end{remark}

Before we start the proof, let's first describe the high-level idea.  To do this, recall that Lieb's form \cite{lieb1978} of the EPI is as follows:  For independent random vectors $X_1,X_2\in \mathcal{P}(\R)$ and any $\lambda\in (0,1)$,  
\begin{align}
h(\sqrt{\lambda} X_1 + \sqrt{1-\lambda} X_2  )\geq \lambda h(X_1)  + (1-\lambda) h(X_2).   \label{eq:introLieb}
\end{align}

Motivated by the similarity between the entropy power inequality and  the {B}runn--{M}inkowski inequality,  Costa and Cover \cite{costa1984similarity} reformulated \eqref{eq:introLieb} as the following concise Gaussian comparison\footnote{The comparison also holds in the multidimensional setting, distinguishing it from the Zamir--Feder inequality.}.
\begin{proposition}[Comparison form of Shannon--Stam inequality]
For independent random variables $X_1, X_2 \in \mathcal{P}(\R)$,  we have
\begin{align}
h(X_1 + X_2)\geq h(Z_1 + Z_2) ,\label{eq:EPIgaussComparison}
\end{align}
where $Z_1,Z_2$ are independent Gaussian random variables with variances chosen so that  $h(Z_i) = h(X_i)$.    
\end{proposition}
To understand how this comes about, observe that a change of variables in \eqref{eq:introLieb} yields the equivalent formulation 
$$
 c h(X_1) + (1-c) h(X_2) + \frac{1}{2}h_2(c) \leq h(X_1 + X_2),\hspace{5mm}\mbox{for all $c\in [0,1]$,}
$$
where $h_2(c):=  - c\log(c) - (1-c)\log(1-c)$  is the binary entropy function.  Since the RHS does not depend on $c$, we may maximize the LHS over $c\in [0,1]$, yielding \eqref{eq:EPIgaussComparison}.  Now, we draw the reader's attention to the formal similarity to \eqref{eq:MainEntropyCouplingInequality}.  Namely, we can apply the same logic to bound
\begin{align}
\sup_{\mathbf{c} \geq 0} \left\{ \sum_{i=1}^k c_i h(X_i) - D_g(\mathbf{c},\mathbf{d},\mathbf{B}) \right\} \leq \max_{X\in \Pi(X_1, \dots, X_k)}\sum_{j=1}^m d_j h(B_j X) . \label{eq:MainEntropyCouplingInequalityToOptimize}
\end{align}
The difficulty encountered is that, unlike $c\mapsto h_2(c)$, the function $\mathbf{c}\mapsto D_g(\mathbf{c},\mathbf{d},\mathbf{B})$ is not   explicit, complicating the optimization problem to be solved.  Nevertheless, the task can be accomplished with all the ingredients we have at hand.  

\begin{proof}[Proof of Theorem \ref{thm:GaussianComparisons}]  We start by noting each maximum is attained due to Proposition \ref{prop:MaxEntropyCouplingExists}.   Now, without loss of generality, we can   assume $\mathbf{d}$ is scaled  so that 
 \begin{align}
 \sum_{j=1}^m d_j \dim(E^j) = 1.\label{eq:normalized}
 \end{align}
 Also, since there are no qualifications on the linear maps in $\mathbf{B}$, a simple rescaling argument reveals that we can assume without loss of generality that $h(X_i)=\frac{\dim(E_i)}{2}\log(2\pi e)$; this will allow us to consider $Z_i\sim N(0,K_i)$ with $\det(K_i)=1$ for each $1\leq i\leq k$.  Thus, by Theorem \ref{thm:FRBLentropy}, we have
\begin{align}
\max_{X\in \Pi(X_1, \dots, X_k)}\sum_{j=1}^m d_j h(B_j X) &\geq \sum_{i=1}^k c_i h(X_i) - D_g(\mathbf{c},\mathbf{d},\mathbf{B})=\frac{1}{2}\log(2\pi e)\sum_{i=1}^k c_i \dim(E_i)   - D_g(\mathbf{c},\mathbf{d},\mathbf{B}) \label{eq:quantityToBound}
\end{align}
for any $\mathbf{c}$.  Define the simplex 
$$A := \left\{\mathbf{c}\geq 0 : \sum_{i=1}^k c_i \dim(E_i) =  \sum_{j=1}^m d_j \dim(E^j) =1 \right\},$$
which is compact  and convex.
 By Theorem \ref{thm:FRBLentropy}, we have $D_g(\mathbf{c},\mathbf{d},\mathbf{B})<\infty$ only if $\mathbf{c}\in A$, so our task in maximizing the RHS of \eqref{eq:quantityToBound} is to compute
$$
\max_{\mathbf{c}\in A}- D_g(\mathbf{c},\mathbf{d},\mathbf{B}) = -\min_{\mathbf{c}\in A} D_g(\mathbf{c},\mathbf{d},\mathbf{B}), 
$$
where the use of $\max$ and $\min$ is justified, since $\mathbf{c} \mapsto D_g(\mathbf{c},\mathbf{d},\mathbf{B})$ is \lsc by Theorem \ref{thm:FunctionalPropertiesDg} and $A$ is compact. For $\mathbf{c}\in A$ and $(K_1,\dots,K_k)\in \prod_{i=1}^k\pd (E_i)$,  define  
$$
F\left(\mathbf{c}, (K_i)_{i=1}^k\right) := \max_{K \in \Pi(K_1, \dots, K_k)}   \sum_{j=1}^m d_j \log \det(B_j K B_j^T) - \sum_{i=1}^k c_i \log \det(K_i),
$$
which is the same as that in \eqref{eq:DgFromF}.   Theorem \ref{thm:FunctionalPropertiesDg} ensures that $F$ satisfies the hypotheses of Theorem \ref{thm:SionGeodesic}.  Thus, by an application of the latter and definition of $D_g(\mathbf{c},\mathbf{d},\mathbf{B})$, we have
\begin{align*}
\max_{\mathbf{c}\in A}- 2D_g(\mathbf{c},\mathbf{d},\mathbf{B}) &= \max_{\mathbf{c}\in A}~~ \inf_{ (K_i)_{i=1}^k \in \prod_{i=1}^k\pd (E_i)} F\left(\mathbf{c}, (K_i)_{i=1}^k\right)\\
&=\!\!\!\inf_{ (K_i)_{i=1}^k \in \prod_{i=1}^k\pd (E_i)} ~\max_{\mathbf{c}\in A} F\left(\mathbf{c}, (K_i)_{i=1}^k\right)\\
&=\!\!\!\inf_{ (K_i)_{i=1}^k \in \prod_{i=1}^k\pd (E_i)} ~\max_{K \in \Pi(K_1, \dots, K_k)}   \sum_{j=1}^m d_j \log \det(B_j K B_j^T)  - \min_{1 \leq i \leq k}\!\!\frac{\log\det(K_i)}{\dim(E_i)}\\
&=\!\!\! \inf_{ \substack{ (K_i)_{i=1}^k \in \prod_{i=1}^k\pd (E_i) :\\ \min_{1\leq i \leq k} \det(K_i) = 1}} ~\max_{K \in \Pi(K_1, \dots, K_k)}   \sum_{j=1}^m d_j \log \det(B_j K B_j^T) ,
\end{align*}
where the last line made use of the observation that the function
$$
(K_i)_{i=1}^k \mapsto \max_{K \in \Pi(K_1, \dots, K_k)}   \sum_{j=1}^m d_j \log \det(B_j K B_j^T)  - \min_{1 \leq i \leq k}\!\!\frac{\log\det(K_i)}{\dim(E_i)}
$$
is invariant to rescaling $(K_i)_{i=1}^k \mapsto (\alpha K_i)_{i=1}^k$ for $\alpha >0$ by \eqref{eq:normalized}.

Now, invoking Theorem \ref{thm:FRdualQuadraticForms}, we have 
\begin{align*}
&   \inf_{ \substack{ (K_i)_{i=1}^k \in \prod_{i=1}^k\pd (E_i) :\\ \min_{1\leq i \leq k} \det(K_i) = 1}} ~\max_{K \in \Pi(K_1, \dots, K_k)}   \sum_{j=1}^m d_j \log \det(B_j K B_j^T)\\
&=  \inf_{ \substack{ (K_i)_{i=1}^k \in \prod_{i=1}^k\pd (E_i) :\\ \min_{1\leq i \leq k} \det(K_i) = 1}} \inf_{(U_i)_{i=1}^k,(V_j)_{j=1}^m}  \left( 
 \sum_{i=1}^k    \langle U_i, K_i\rangle_{\HS} - \sum_{j=1}^m d_j \log \det V_j\right),
\end{align*}
where the second infimum  is over all $U_i\in \pd (E_i),1\leq i\leq k$ and $V_j\in \pd (E^j), 1\leq j\leq m$ satisfying  
\begin{align*}
\sum_{j=1}^m d_j B_j^T V_j B_j \leq \operatorname{diag}(  U_1, \dots,   U_k).  
\end{align*}
Written in this way, it evidently suffices to consider $\det(K_i) = 1$ for all $1\leq i\leq k$ in the last line, 
so we conclude
\begin{align}
\max_{\mathbf{c}\in A}- 2D_g(\mathbf{c},\mathbf{d},\mathbf{B})  =  \inf_{ \substack{ (K_i)_{i=1}^k \in \prod_{i=1}^k\pd (E_i) :\\ \det(K_i) = 1, 1\leq i\leq k}} ~\max_{K \in \Pi(K_1, \dots, K_k)}   \sum_{j=1}^m d_j \log \det(B_j K B_j^T). \label{matrixIdent}
\end{align}
Now, let $\mathbf{c^*} \in \arg\min_{\mathbf{c}\in A} D_g(\mathbf{c},\mathbf{d},\mathbf{B})$.  By \eqref{eq:quantityToBound} and \eqref{eq:normalized}, we have
\begin{align}
\max_{X\in \Pi(X_1, \dots, X_k)}\sum_{j=1}^m d_j h(B_j X) &\geq \frac{1}{2}\log(2\pi e)    - D_g(\mathbf{c^*},\mathbf{d},\mathbf{B}). \label{eq:DependsOnExtremizability}
\end{align}
If the LHS of \eqref{eq:DependsOnExtremizability} is equal to $-\infty$,  then it is easy to see that one of the $B_j$'s must fail to be surjective. Indeed, suppose each $B_j$ is surjective and  factor $B_j = R_j Q_j$, where $Q_j$ has orthonormal rows and $R_j$ is full rank.  Letting $Q^{\perp}_j$ denote the matrix with orthonormal rows and rowspace equal to the orthogonal complement of the rowspace of $Q_j$, for the independent coupling $X$ we have 
$$
\sum_{i=1}^k h(X_i) = h(X) =h(Q_j X, Q_j^{\perp} X) \leq  h(Q_j X) + h(Q_j^{\perp} X).%
$$
Since $h(Q_j^{\perp} X)$ is bounded from above due to finiteness of second moments and the LHS is finite by assumption, $h(Q_j X)$ is finite, and so is $h(B_j X)$. 
Therefore, \eqref{eq:maxEntComparison} holds trivially if the LHS of \eqref{eq:DependsOnExtremizability} is equal to $-\infty$, so we  assume henceforth that the LHS of \eqref{eq:DependsOnExtremizability}  is finite.  If $(\mathbf{c^*},\mathbf{d},\mathbf{B})$ is extremizable, then by Theorem \ref{thm:extImpliesGext} and \eqref{matrixIdent},     there exist Gaussians $Z^*_i\sim N(0,K_i)$ with $\det(K_i)=1$ such that
\begin{align*}
\max_{X\in \Pi(X_1, \dots, X_k)}\sum_{j=1}^m d_j h(B_j X) &\geq \frac{1}{2}\log(2\pi e)    - D_g(\mathbf{c^*},\mathbf{d},\mathbf{B})\\
&=\max_{Z\in \Pi(Z^*_1, \dots, Z^*_k)}\sum_{j=1}^m d_j h(B_j Z),
\end{align*}
where we used the  identity  $\frac{1}{2}\log(2\pi e) = \sum_{i=1}^k c_i^* h(X_i) =  \sum_{i=1}^k c_i^* h(Z^*_i)$.  On the other hand, if $(\mathbf{c^*},\mathbf{d},\mathbf{B})$ is not extremizable, then we have strict inequality in \eqref{eq:DependsOnExtremizability}, and it follows by \eqref{matrixIdent} that there are Gaussians $Z_i\sim N(0,K_i)$ with $\det(K_i)=1$ such that \eqref{eq:maxEntComparison} holds (with strict inequality, in fact).
\end{proof}

\subsection{Proof of Theorem \ref{thm:GaussianComparisonConstrained}}

With   Theorem \ref{thm:GaussianComparisons} at our disposal, it is a straightforward matter to self-strengthen it to produce Theorem \ref{thm:GaussianComparisonConstrained}. 

First, observe that lower semicontinuity of relative entropy implies $X\in \Pi(X_1, \dots, X_k) \mapsto I_S(X)$ is weakly lower semicontinuous, and therefore $\Pi(X_1, \dots, X_k;\nu)$ is a compact subset of $\Pi(X_1, \dots, X_k)$ when equipped with the weak topology.  Hence, repeating the argument in the Proposition \ref{prop:MaxEntropyCouplingExists}, we find that each  maximum is achieved the statement of the Theorem.

Now, by the method of Lagrange multipliers, 
\begin{align*}
\max_{X\in \Pi(X_1, \dots, X_k; \nu)} \sum_{j=1}^m d_j h(B_j X) &= \max_{X\in \Pi(X_1, \dots, X_k)} ~\inf_{\lambda\geq 0} \left( \sum_{j=1}^m d_j h(B_j X) - \sum_{S: \nu(S)<\infty } \lambda(S) (I_S(X) - \nu(S))\right) \\
&= \inf_{\lambda\geq 0}~\max_{X\in \Pi(X_1, \dots, X_k)}  \underbrace{\left( \sum_{j=1}^m d_j h(B_j X) -\sum_{ S: \nu(S)<\infty }  \lambda(S) (I_S(X) - \nu(S))\right)}_{=:G(\lambda, X)} ,
\end{align*}
where the infimum is over functions $\lambda : 2^{\{1,\dots, k\}} \to [0,+\infty)$. The exchange of   $\max$ and $\inf$ follows by an application of the classical Sion minimax theorem.  Indeed, for any fixed $X\in \Pi(X_1, \dots, X_k )$, the function $\lambda \mapsto G(\lambda, X)$ is linear in $\lambda$.  On the other hand, $\Pi(X_1, \dots, X_k)$ is a convex  subset of $\mathcal{P}(E_0)$ that is compact with respect to the weak topology. For fixed $\lambda\geq 0$, the functional $X \mapsto G(\lambda, X)$ is concave upper semicontinuous on $\Pi(X_1, \dots, X_k)$ by  concavity of entropy and Lemma \ref{lem:WeakSemicontH}.     

Using the   definition of $I_S$, for any $\lambda\geq 0$, Theorem \ref{thm:GaussianComparisons} applies to give existence of Gaussian $(Z_i)_{i=1}^k$ satisfying 
\begin{align*}
&\max_{X\in \Pi(X_1, \dots, X_k)}   \left( \sum_{j=1}^m d_j h(B_j X) -  \sum_{ S: \nu(S)<\infty }  \lambda(S) (I_S(X) - \nu(S))\right) \\
&\geq  \max_{Z\in \Pi(Z_1, \dots, Z_k )} \left( \sum_{j=1}^m d_j h(B_j Z)
-  \sum_{ S: \nu(S)<\infty }  \lambda(S) (I_S(Z) - \nu(S))\right) \\
&\geq  \max_{Z\in \Pi(Z_1, \dots, Z_k;\nu )}  \sum_{j=1}^m d_j h(B_j Z).
\end{align*}
The last inequality follows since we are taking the maximum over a smaller set and because $\lambda\geq 0$.  This proves the theorem.

\section{Application: constrained multi-marginal inequalities} \label{sec:multimarginal}
In this section, we introduce a constrained version of the multi-marginal inequality considered in \eqref{eq:MainEntropyCouplingInequality} and demonstrate how the results transfer almost immediately with the help of Theorem \ref{thm:GaussianComparisonConstrained}.  

Fix a datum $(\mathbf{c},\mathbf{d},\mathbf{B})$.  For a constraint function $\nu: 2^{\{1,\dots, k\}}\to [0,+\infty]$, let $D(\mathbf{c},\mathbf{d},\mathbf{B};\nu)$ denote the smallest constant $D$ such that the inequality 
\begin{align}
\sum_{i=1}^k c_i h(X_i) \leq \max_{X\in \Pi(X_1, \dots, X_k;\nu)}\sum_{j=1}^m d_j h(B_j X) + D \label{eq:multimarginalConst}
\end{align}
holds for all choices of $X_i \in \mathcal{P}(E_i)$, $1\leq i\leq k$.  Call $(\mathbf{c},\mathbf{d},\mathbf{B};\nu)$  {\bf extremizable} if there are $X_i \in \mathcal{P}(E_i)$, $1\leq i\leq k$ which achieve equality in \eqref{eq:multimarginalConst} with $D = D(\mathbf{c},\mathbf{d},\mathbf{B};\nu)$.   Similarly, let $D_g(\mathbf{c},\mathbf{d},\mathbf{B};\nu)$ denote the smallest constant $D$ such that \eqref{eq:multimarginalConst} holds for all Gaussian $X_i \in \mathcal{G}(E_i)$, $1\leq i\leq k$, and call $(\mathbf{c},\mathbf{d},\mathbf{B};\nu)$  {\bf Gaussian-extremizable} if there are $X_i \in \mathcal{G}(E_i)$, $1\leq i\leq k$ which achieve equality in \eqref{eq:multimarginalConst} with $D = D_g(\mathbf{c},\mathbf{d},\mathbf{B};\nu)$.

The following generalizes Theorem \ref{thm:FRBLentropy} and \ref{thm:extImpliesGext} to the correlation-constrained setting.
\begin{theorem}\label{thm:constrainedFRBLentropy}
For any datum $(\mathbf{c},\mathbf{d},\mathbf{B})$ and constraint function $\nu$,  
\begin{enumerate}[(i)]
\item $D(\mathbf{c},\mathbf{d},\mathbf{B};\nu)=D_g(\mathbf{c},\mathbf{d},\mathbf{B};\nu)$; 
\item $(\mathbf{c},\mathbf{d},\mathbf{B};\nu)$ is extremizable if and only if it is Gaussian-extremizable; and
\item $D_g(\mathbf{c},\mathbf{d},\mathbf{B};\nu)$ is finite if and only if the scaling condition \eqref{eq:ScalingCond} and the dimension condition \eqref{eq:DimCond} hold.
\end{enumerate}
\end{theorem}
\begin{proof}
For any $X_i\in \mathcal{P}(E_i)$ and any $\mathbf{c}$, an application of Theorem \ref{thm:GaussianComparisonConstrained} ensures existence of $Z_i \in \mathcal{G}(E_i)$ with $h(Z_i)=h(X_i)$ satisfying 
\begin{align*}
&\sum_{i=1}^k c_i h(X_i) - \max_{X\in \Pi(X_1, \dots, X_k;\nu)}\sum_{j=1}^m d_j h(B_j X)\\
&\leq \sum_{i=1}^k c_i h(Z_i) - \max_{Z\in \Pi(Z_1, \dots, Z_k;\nu)}\sum_{j=1}^m d_j h(B_j Z)  \leq D_g(\mathbf{c},\mathbf{d},\mathbf{B};\nu),
\end{align*}
where the final inequality follows by definition of $D_g$.  This establishes both (i) and (iii).  As for finiteness, observe that definitions imply
\begin{align}
D_g(\mathbf{c},\mathbf{d},\mathbf{B}) \equiv D_g(\mathbf{c},\mathbf{d},\mathbf{B}; +\infty) \leq D_g(\mathbf{c},\mathbf{d},\mathbf{B};\nu)\leq D_g(\mathbf{c},\mathbf{d},\mathbf{B};0)\label{eq:finitenessIneq}
\end{align}
for any $\nu$.  Now, for any $K\in \Pi(K_1, \dots, K_k)$ with $K_i\in \pd (E_i)$, $1\leq i \leq k$, observe that  
$$K\leq k \operatorname{diag}(K_1, \dots, K_k).$$  Indeed, for $Z\sim N(0,K)$ and $u=(u_1,\dots, u_k) \in E_0$, Jensen's inequality yields
$$
u^T K u = \EE|u^T Z|^2 \leq k \sum_{i=1}^k \EE |u_i^T Z_i |^2 =   k u^T \operatorname{diag}(K_1, \dots, K_k) u.
$$
This implies, for Gaussian $(Z_i)_{i=1}^k$, that
$$
\max_{Z \in \Pi(Z_1,\dots, Z_k)} \sum_{j=1}^m d_j h(B_j Z) \leq \sum_{j=1}^m d_j h(B_j Z^{\mathrm{ind}}) + \log(k) \sum_{j=1}^m d_j\dim(E^j),
$$
where $Z^{\mathrm{ind}}$ denotes  the independent coupling of the $Z_i$'s. Thus,
$$
D_g(\mathbf{c},\mathbf{d},\mathbf{B};0) \leq D_g(\mathbf{c},\mathbf{d},\mathbf{B})+ \log(k) \sum_{j=1}^m d_j\dim(E^j),
$$
so that finiteness of $D_g(\mathbf{c},\mathbf{d},\mathbf{B};\nu)$ is equivalent to finiteness of $D_g(\mathbf{c},\mathbf{d},\mathbf{B})$ by \eqref{eq:finitenessIneq}. Invoking Theorem \ref{thm:FRBLentropy} completes the proof.
\end{proof}

When $\nu \equiv 0$, then the only allowable coupling in \eqref{eq:multimarginalConst} is the independent one.  Thus, we recover the main results of Anantharam, Jog and Nair \cite[Theorems 3 \& 4]{anantharam2019unifying}, which simultaneously capture the entropic Brascamp--Lieb inequalities and the EPI.

When $\nu \equiv +\infty$, then we immediately recover Theorems \ref{thm:FRBLentropy} and \ref{thm:extImpliesGext}.  Of note, we recall from \cite{liu2018forward, CourtadeLiu21} that, by extending the duality for the Brascamp--Lieb inequalities \cite{carlen2009subadditivity}, Theorem \ref{thm:FRBLentropy}  has the following equivalent functional form.
\begin{theorem}\label{thm:FRBLfunctional}
Fix a datum $(\mathbf{c},\mathbf{d},\mathbf{B})$.  If measurable functions $f_i : E_i \to \R^+$,  $1\leq i \leq k$ and $g_j : E^j \to \R^+$,  $1\leq j\leq m$  satisfy 
\begin{align}
\prod_{i=1}^k f_i^{c_i}(\pi_{E_i}(x)) \leq \prod_{j=1}^m g_j^{d_j}\left( B_j x   \right)\hspace{1cm}\forall x\in E_0,\label{eq:majorization}
\end{align}
then 
\begin{align}
\prod_{i=1}^k \left( \int_{E_i} f_i  \right)^{c_i} \leq e^{ D_g(\mathbf{c},\mathbf{d},\mathbf{B}) } \prod_{j=1}^m \left( \int_{E^j} g_j  \right)^{d_j}.\label{eq:frblFunctional}
\end{align}
 Moreover, the constant $D_g(\mathbf{c},\mathbf{d},\mathbf{B})$  is best possible. 
\end{theorem}
%


By a suitable choice of datum $(\mathbf{c},\mathbf{d},\mathbf{B})$, this implies many geometric inequalities such as the Brascamp--Lieb inequalities \cite{brascamp1974general, brascamp1976best, lieb1990gaussian} (which include, e.g.,  H\"older's inequality, the sharp Young  inequality, the Loomis--Whitney inequalities), the Barthe inequalities \cite{barthe1998reverse} (which include, e.g.,  the Pr\'ekopa--Leindler inequality, Ball's inequality \cite{ball1989volumes}), the sharp reverse Young inequality \cite{brascamp1976best},  the Chen--Dafnis--Paouris inequalities \cite{chen2015improved}, and a form of the Barthe--Wolff inequalities \cite{barthe2018positive}.  Readers are referred to \cite{CourtadeLiu21} for a more detailed account of these implications and further references.  The survey by Gardner also gives a clear depiction of the hierarchy implied by  the Brascamp--Lieb and Barthe inequalities \cite[Fig. 1]{gardner2002brunn}.

We remark that, while Theorem \ref{thm:FRBLentropy} admits the equivalent functional form given above, there is no obvious functional equivalent when $\nu$ induces nontrivial correlation constraints.  In particular, the comparison \eqref{eq:maxEntComparisonConstrained} seems to be most naturally expressed in the language of entropies (even in the unconstrained case).

\section{Application: Gaussian saddle point}\label{sec:saddle}

The EPI has been successfully applied many times to prove coding theorems, particularly in the field of network information theory.  However, it also provides the essential ingredient in establishing that a certain mutual information game admits a saddle point (see  \cite{pinsker1956calculation, Ihara}, and also \cite[Problem 9.21]{coverThomas}).  Namely, for numbers $P,N\geq 0$, we have 
\begin{align}
\sup_{P_X: \EE|X|^2\leq P} ~\inf_{P_Z: \EE|Z|^2\leq N} I(X;X+Z) =  \inf_{P_Z: \EE|Z|^2\leq N} ~\sup_{P_X: \EE|X|^2\leq P} I(X;X+Z) , \notag %
\end{align}
where the $\sup$ (resp.\ $\inf$) is over  $X\sim P_X\in \mathcal{P}(\mathbb{R}^n)$ such that $\EE|X|^2\leq P$  (resp.\ $Z\sim P_Z\in \mathcal{P}(\mathbb{R}^n)$ such that $\EE|Z|^2\leq N$), and the mutual information is computed under the assumption that $X\sim P_X$ and $Z\sim P_X$ are independent.  It turns out that the game admits a Gaussian saddle point, which together with Shannon's capacity theorem,   implies that worst-case additive noise is Gaussian.

In this section, we extend this saddle point property to a game with payoff  given by
$$
G_{\zeta}(P_X, P_Z) := \sup_{ \substack{ (X,Z)\in\Pi(P_X,P_Z):\\ I(X;Z)\leq \zeta}} I(X; X+Z),
$$
for a parameter $\zeta\geq 0$, where the supremum is over couplings $(X,Z)$ with given marginals $X\sim P_X$ and $Z\sim P_Z$. 
Of course, by taking $\zeta = 0$, we will recover the classical saddle-point result above.   This may be interpreted as a game where the signal and noise players fix their strategies $P_X$ and $P_Z$, but the signal player has the benefit during game-play of adapting their transmission using side information obtained about the noise player's action.
\begin{theorem}\label{thm:SaddlePt} For $0< P,N < \infty$ and $\zeta\geq 0$,
\begin{align*}
&\sup_{P_{X}: \EE|X|^2\leq P} ~\inf_{P_{Z}: \EE|Z|^2\leq N} G_{\zeta}(P_X, P_Z) =  \inf_{P_{Z}: \EE|Z|^2\leq N} ~\sup_{P_{X}: \EE|X|^2\leq P} G_{\zeta}(P_X, P_Z) .
\end{align*}
Moreover, $P_X =  N\left(0,\tfrac{P}{n}\id_{\mathbb{R}^n}\right)$ and $P_Z =  N\left(0,\tfrac{N}{n}\id_{\mathbb{R}^n}\right)$ is a saddle point.
\end{theorem}

\begin{proof}[Proof of Theorem \ref{thm:SaddlePt}]  In a slight abuse of notation, we will write $ \Pi(X_1, X_2; \zeta)$  to denote couplings of $X_1,X_2$ satisfying $I(X_1;X_2)\leq \zeta$. 

Let $X$ and $Z$ be a random variables with finite variance, and let $X^*,Z^*$ be centered isotropic Gaussians with $\EE|X^*|^2 = \EE|X|^2$ and $\EE|Z^*|^2 = \EE|Z|^2$.   Now, observe that Theorem \ref{thm:depEPI} implies
\begin{align*}
\max_{  \Pi(X^*, Z; \zeta) } \left( h(X^*+ Z) - h(Z) \right) &\geq  \frac{n}{2}\log\left( 1 + \frac{N(X^*)}{N(Z)}   + 2 \sqrt{ (1 - e^{- \frac{2 \zeta}{n} }) \frac{N(X^*)}{N(Z)}  }\right)\\
&\geq  \frac{n}{2}\log\left( 1 + \frac{N(X^*)}{N(Z^*)}   + 2 \sqrt{ (1 - e^{- \frac{2 \zeta}{n}}) \frac{N(X^*)}{N(Z^*)}  }\right)\\
&=\max_{  \Pi(X^*, Z^*; \zeta) } \left( h(X^*+ Z^*) - h(Z^*) \right), 
\end{align*}
where the second inequality follows since $h(Z) \leq h(Z^*)$, and the last equality follows by the equality conditions in Theorem \ref{thm:depEPI}.  In particular, this gives
\begin{align}
\sup_{  \Pi(X^*, Z; \zeta) } I(X^*; X^*+ Z) &= \sup_{  \Pi(X^*, Z; \zeta) } \left( h(X^*+ Z) - h(Z) + I(X^*; Z)\right) \notag\\
&= \sup_{  \Pi(X^*, Z; \zeta) } \left( h(X^*+ Z) - h(Z) \right)+ \zeta \label{secondEquality} \\
&\geq \sup_{  \Pi(X^*, Z^*; \zeta) } \left( h(X^*+ Z^*) - h(Z^*) \right)+ \zeta \label{applyCor}\\
&=\sup_{  \Pi(X^*, Z^*; \zeta) } I(X^*; X^*+ Z^*), \notag
\end{align}
where  \eqref{secondEquality} can be justified using the supremum\footnote{This sounds obvious, but we don't know of a simple argument to justify the assertion.  A proof  is given in Proposition \ref{prop:rearrangementArgument}.}, and \eqref{applyCor} follows from the previous computation. 
For any pair $(X,Z^*)$, couple $(X^*, Z^*)$ to have the same covariance.  By the max-entropy property of Gaussians, $I(X^*; Z^*)\leq I(X;Z^*)$ and $h(X+ Z^*) \leq h(X^*+ Z^*)$.  As a result, we have
\begin{align*}
\sup_{  \Pi(X, Z^*; \zeta) }\!\! I(X; X+ Z^*) \leq \sup_{  \Pi(X^*, Z^*; \zeta) } \!\!\! I(X^*; X^*+ Z^*) \leq \sup_{  \Pi(X^*, Z; \zeta) } \!\!I(X^*; X^*+ Z) .
\end{align*}
This implies 
\begin{align*}
\inf_{P_{Z}: \EE|Z|^2\leq N} ~\sup_{P_{X}: \EE|X|^2\leq P} G_{\zeta}(P_X, P_Z)  \leq &\sup_{P_{X}: \EE|X|^2\leq P} ~\inf_{P_{Z}: \EE|Z|^2\leq N} G_{\zeta}(P_X, P_Z), 
\end{align*}
and the reverse direction follows by  the max-min inequality.  The fact that the asserted distributions coincide with the saddle point subject to the constraints follows by direct computation.
\end{proof}

We now tie up loose ends by justifying \eqref{secondEquality}, which is an easy consequence of the proposition below. 
\begin{proposition}\label{prop:rearrangementArgument}
Let $X\sim N(0,\id_{\mathbb{R}^n})$ and   $Z \in \mathcal{P}(\mathbb{R}^n)$  be jointly distributed with $I(X;Z) \leq \zeta < + \infty$.  For any $\epsilon>0$, there is a coupling $(X',Z') \in \Pi(X,Z)$ with $h(X'+Z') \geq h(X+Z)-\epsilon$ and $I(X';Z')=\zeta$.
\end{proposition}
\begin{proof}  We'll work in dimension $n=1$ for simplicity of exposition.  It suffices to establish existence of $(X',Z') \in \Pi(X,Z)$ with $h(X'+Z') \geq h(X+Z)-\epsilon$ and $\zeta \leq I(X';Z') < +\infty$.  Indeed, if there is such $(X',Z')$, then we can  let $\pi_0$ denote the joint distribution of $(X,Z)$ and $\pi_1$ denote the joint distribution of $(X',Z')$.  For $\theta\in[0,1]$ define the mixture 
$$
\pi_{\theta}= (1-\theta) \pi_0 + \theta \pi_1. 
$$
Evidently, $\pi_{\theta}\in \Pi(X,Z)$ for all $\theta\in[0,1]$.  For $(X^{(\theta)}, Z^{(\theta)})\sim \pi_{\theta}$, concavity of entropy gives
\begin{align*}
h(X^{(\theta)}+ Z^{(\theta)}) &\geq (1-\theta)h(X+Z)+ \theta h(X'+Z')  \geq h(X+Z) - \epsilon.
\end{align*}
Now, convexity of relative entropy ensures that $\theta \mapsto I(X^{(\theta)}; Z^{(\theta)})$ is continuous on $(0,1)$.  Weak lower semicontinuity of mutual information together with finiteness of $I(X'; Z')$ establishes continuity at the endpoints, so that the above mapping is continuous on $[0,1]$.  As a result, the intermediate value theorem ensures there is some $\theta \in [0,1]$ such that $h(X^{(\theta)}+ Z^{(\theta)}) \geq  h(X+Z) - \epsilon$ 
and $I(X^{(\theta)}; Z^{(\theta)})=\zeta$. 

Toward establishing the above ansatz, fix $\epsilon>0$, and consider the interval $I:=(-\epsilon,\epsilon]$.  Define $p(\epsilon):= \Pr\{X\in I\}$, and note that $p(\epsilon) = \Theta(\epsilon)$ since $X$ is assumed Gaussian. For fixed parameters $n\geq1$ and $\epsilon$, we'll rearrange the joint distribution of $(X, Z)$ on the event $\{X\in I\}$.  To this end, consider two partitions
$$
-\epsilon = t_0 < t_1 < \cdots < t_n = \epsilon
$$
and
$$
-\infty = s_0 < s_1 < \cdots < s_n = +\infty
$$
such that 
$$
\Pr\{X \in (t_{i-1}, t_i]|X\in I\} = \Pr\{Z \in (s_{i-1}, s_i]|X\in I\} = \frac{1}{n}, ~~1\leq i  \leq n.
$$
This is always possible since  $X$ and $Z$ are (marginally) continuous random variables.  We now define a random variable $Z_n$, jointly distributed with $X$, by rearranging the distribution of $(X,Z)$ as follows.  On the event $\{X\notin I\}$, we let $Z_n = Z$.  Conditioned on the event $\{X\in I\}$, we let the joint density of $(X,Z_n)$ be supported on the union of rectangles $R:=\cup_{i=1}^n (t_{i-1}, t_{i}]\times (s_{i-1}, s_{i}]$, given explicitly by
$$
f_{X,Z_n|X\in I}(x,z) = n f_{X|X\in I}(x) f_{Z|X\in I}(z) 1_{\{(x,z) \in R\}}.
$$
This is well-defined since the conditional densities $f_{X|X\in I}, f_{Z|X\in I}$ exist by marginal continuity of $X$ and $Z$, and the fact that $\Pr\{X\in I\}>0$.

Observe that this rearrangement preserves marginals,  so $(X,Z_n)\in \Pi(X,Z)$.   Further, note that $I(X; Z_n | X\in I) = n$ by construction, therefore 
\begin{align*}
I(X; Z_n) &= p(\epsilon) I(X; Z_n | X\in I) +  (1-p(\epsilon)) I(X; Z_n | X\notin I) + I(Z; 1_{\{X\in I\}})\\
&\leq p(\epsilon) n + I(X; Z) + O(\epsilon \log \epsilon).
\end{align*}
By nonnegativity of mutual information, the first identity above also implies
$$I(X; Z_n) \geq p(\epsilon) I(X; Z_n | X\in I) = p(\epsilon) n.$$  %
Since $I(X;Z)$ is finite by assumption, the combination of the above estimates imply
\begin{align}
I(X; Z_n) = \Theta(\epsilon n), \label{eq:IXZn}
\end{align}
where the asymptotics are understood in the sense that $\epsilon>0$ is fixed and $n$ allowed to increase.

For $x\in I$, let $k(x)$ denote the integer $k\in \{1,\dots, n\}$ such that $x \in (t_{k-1}, t_k]$. Observe that, conditioned on $\{X\in I\}$, the index $k(X)$ is almost surely equal to a function of $X+Z_n$.  This follows since for any $c\in \mathbb{R}$,  the line $\{(x,z) : x + z = c\} \subset \mathbb{R}^2$ intersects a  unique rectangle of the form
$$
(t_{i-1}, t_{i}]\times (s_{i-1}, s_{i}].
$$
Conditioned on $\{X\in I\}$, the distribution of $(X,Z_n)$ is supported on such rectangles by construction, so the claim follows.

With the above observation together with the fact that  $X$ and $Z_n$ are conditionally independent given $\{k(X), X\in I\}$ by construction,   we have 
\begin{align*}
h(X+ Z_n | X\in I) &= h(X+ Z_n | k(X), X\in I) + I(X+ Z_n ; k(X)| X\in I) \\
&\geq h(X | k(X), X\in I) + I(X; k(X)| X\in I) \\
&= h(X |  X\in I) .
\end{align*}
In particular, 
\begin{align*}
h(X+Z_n) &\geq (1-p(\epsilon))h(X+Z_n|X\notin I) +  p(\epsilon) h(X+Z_n |X\in I) \\
&\geq (1-p(\epsilon))h(X+Z_n |X\notin I) +  p(\epsilon) h(X  |X \in I) \\
&= (1-p(\epsilon))h(X+Z |X\notin I) +  O(\epsilon \log \epsilon),
\end{align*}
where the first line follows since conditioning reduces entropy, and the last line follows since $X$ is nearly uniform on $I$ for $\epsilon$ small. 

Now, upper bounding entropy in terms of second moments, we have
\begin{align*}
h(X+Z |X\in I)&\leq \frac{1}{2}\log\left(2\pi e\EE[(X+Z)^2|X\in I] \right) \\
&\leq \frac{1}{2}\log\left(4\pi e( \EE[X^2|X\in I] +\EE[Z^2|X\in I] )  \right)\\
&\leq \frac{1}{2}\log\left(\frac{4\pi e}{p(\epsilon)}( \EE[X^2 ] +\EE[Z^2 ] )  \right).
\end{align*}
So, by finiteness of second moments, 
\begin{align*}
p(\epsilon) h(X+Z |X\in I)&\leq O(\epsilon \log \epsilon).
\end{align*}
Since $I(X+Z; 1_{\{X\in I\}})\leq H(1_{\{X\in I\}}) = O(\epsilon \log \epsilon)$, we put everything together to find 
\begin{align*}
h(X+Z) &= h(X+Z|1_{\{X\in I\}} ) +I(X+Z; 1_{\{X\in I\}}) \\
&\leq (1-p(\epsilon))h(X+Z |X\notin I) + O(\epsilon \log \epsilon)\\
&\leq h(X+Z_n) + O(\epsilon \log \epsilon).
\end{align*}
Combining with \eqref{eq:IXZn} establishes the ansatz, and completes the proof.
\end{proof}

\section{Closing Remarks}\label{sec:closing}
Through the sequence of applications given, we have hopefully convinced the reader that Theorem \ref{thm:GaussianComparisonConstrained} offers significant generality and flexibility.  This flexibility allows for easy modification.  For example, the results hold verbatim over the complex field instead of the real field.  Indeed, this follows directly by considering complex-valued random variables as two-dimensional real random variables, and using the matrix representation of complex numbers to implement the linear transformations\footnote{With a bit more work, which we do not detail here, it can be shown  in the complex setting that it suffices to consider circularly symmetric complex Gaussians in the lower bound of the comparison \eqref{eq:maxEntComparisonConstrained}.}.  By the same logic, the same comparison of Theorem \ref{thm:GaussianComparisonConstrained} can be seen to hold for other matrix algebras over $\mathbb{R}$.

Having said all this, we make no assertion that Theorem \ref{thm:GaussianComparisonConstrained} is a grand unification of all entropy inequalities on Euclidean space.  Indeed, there are several important examples of inequalities that are not obviously subsumed.  Results in \cite{LiuViswanath, GengNair, CourtadeStrongEPI} provide representative examples.  We concede that there may be some clever application of Theorem \ref{thm:GaussianComparisonConstrained} that can   recover some of these results, %
but we do not know of one at the time of this writing.  Thus, at the moment, it seems that Theorem \ref{thm:GaussianComparisonConstrained} may     be another piece in a larger puzzle still wanting to be put together.

\subsection*{Acknowledgement} 
 This work was supported in part by NSF grant CCF-1750430 (CAREER).

\end{document}